\documentclass[journal,12pt,onecolumn,draftclsnofoot,]{IEEEtran}
\usepackage{cite}
\usepackage{amsmath,amssymb,amsfonts}
\usepackage{algorithmic}
\usepackage{graphicx}
\usepackage{textcomp}
\makeatletter
\renewcommand\@noitemerr{%
  \@latex@warning{Empty `thebibliography' environment}%
}
\makeatother

\usepackage{amsthm}
\newtheorem{thm}{Theorem}
\newtheorem{rem}{Remark}%
\newtheorem{cor}{Corollary}%
\newtheorem{lem}{Lemma}%
\newtheorem{assum}{Assumption}%
\newtheorem{defn}{Definition}%

\usepackage{enumitem}

\usepackage{hyperref} 
\hypersetup{colorlinks=true,colorlinks,linkcolor={blue},citecolor={blue},urlcolor={blue}} 
  
\usepackage{tikz}
\usetikzlibrary{shapes,arrows,positioning,calc}

\begin{document}
\title{Input-to-State Stability of Reset-Integral Sliding Mode Control for Linear Systems}
\author{Xinxin Zhang, \IEEEmembership{Member, IEEE}, and Leonid Freidovich, \IEEEmembership{Senior Member, IEEE}

\thanks{Xinxin Zhang (e-mail: xinxin.zhang@umu.se) and Leonid Freidovich  (e-mail: leonid.freidovich@umu.se) are with Department of Applied Physics and Electronics, Umeå University, Håken Gullessons väg 20, Teknikhuset, 901 87 Umeå. This work is supported by the Kempe Foundation, Sweden, under project number JCSMK24-523.}
}

\maketitle 

\begin{abstract}
This work presents a stability analysis of a hybrid control system integrating a reset controller (RC) featuring a single reset state with an integral sliding-mode controller (ISMC). It is shown that the reachability of the sliding surface is decoupled from the nominal reset mechanism. This decoupling property enables a Lyapunov-based stability analysis, demonstrating that the closed-loop RC-ISMC system achieves input-to-state stability (ISS) and uniform ultimate boundedness (UUB) in the presence of exogenous references and disturbances. Furthermore, global asymptotic stability (GAS) is recovered in the unperturbed regulation scenario. A numerical case study illustrates the theoretical results.
\end{abstract}



\section{Introduction}
\label{sec:intro}  
Reset control systems are a class of hybrid systems that integrate linear continuous-time flow dynamics with a discrete reset mechanism modeled as state jumps \cite{Banos2011}. Well known for their ability to overcome linear control limitations, such as Bode's integral theorem, reset controllers can achieve enhanced transient or steady-state performance and have been widely applied in mechatronics \cite{wu2007reset, Clegg1958, Banos2011}. 

However, the nonlinearity of the reset mechanism can degrade robustness, leaving the closed-loop system sensitive to exogenous perturbations. This limitation motivates the integration of reset control with integral sliding-mode control (ISMC) to reject matched disturbances \cite{rubagotti2011integral}. By enforcing the sliding regime from the initial instant, ISMC eliminates the reaching phase, ensuring matched disturbance rejection from the onset of the system response \cite{Utkin1996, Shtessel2014}.

Nevertheless, it remains unclear whether integrating an ISMC layer compromises the stability of the reset controller. While ISS criteria for reset controllers \cite{beker2004fundamental, Banos2011} and for impulsive systems \cite{bainov1989systems, khalil2002nonlinear} are well established, the stability analysis of the reset control–integral sliding-mode control (RC-ISMC) system remains an open problem.

This open problem is further complicated by a practical design choice in the ISMC layer. Classical ISMC assumes ideal discontinuous switching to instantaneously reject matched disturbances \cite{Shtessel2014}; to avoid the resulting chattering, this work instead adopts a saturation-based boundary layer \cite{Slotine}. Accordingly, the stability analysis of the RC-ISMC architecture presents two primary challenges: (i) establishing finite-time reachability of the boundary layer during continuous flow and ensuring its invariance across discrete jumps, and (ii) preserving the stability of the nominal reset system under the ISMC law. To address these challenges, the main contributions of this work are as follows:
\begin{itemize}
    \item We prove that the saturation-based ISMC guarantees finite-time reachability of a boundary layer around the sliding surface.
    \item By using a quadratic ISS Lyapunov function, we establish that the closed-loop RC-ISMC system achieves ISS and uniform ultimate boundedness (UUB) in the presence of exogenous references and disturbances.
    \item A numerical case study is provided to illustrate the applicability of the theoretical findings.
\end{itemize}

The remainder of this study is organized as follows. Section \ref{sec:prel} introduces the plant model and the reset controller. Section \ref{sec:method} formulates the RC-ISMC law and presents the ISS analysis, establishing the ultimate bounds. Section \ref{sec:results} presents the numerical case study, and Section \ref{sec:conc} concludes the study with a summary of limitations and future work.

\section{Preliminaries}
\label{sec:prel}

\subsection{System Description}
Consider the linear time-invariant (LTI) plant
\begin{equation}
\label{eq:state_space_system}
\begin{cases}
    \dot{x}(t) = Ax(t) + B\bigl(u(t) + d(t)\bigr) + w(t), \\ 
    u(t) = u_{ff}(t) + u_r(t) + u_{sm}(t), \\
    y(t) = Cx(t),
\end{cases}
\end{equation}
where $x(t) \in \mathbb{R}^n$ is the measurable state, $u(t) \in \mathbb{R}$ is the control input, and $y(t) \in \mathbb{R} $ is the measured output. The matrices $A \in \mathbb{R}^{n \times n}$, $B \in \mathbb{R}^{n \times 1}$, and $C \in \mathbb{R}^{1 \times n}$ are constant, and the pair $(A, B)$ is controllable. The terms $d(t) \in \mathbb{R}$ and $w(t) \in \mathbb{R}^n$ denote the exogenous matched and unmatched disturbances, respectively. The control law $u(t)$ consists of a feedforward term $u_{ff}(t)$, a reset control law $u_r(t)$ (detailed in Section \ref{sec:nominal_reset_system}), and an ISMC law $u_{sm}(t)$ (in Section \ref{subsec: ismc}).

Let $r_s(t) \in \mathbb{R}^n$ denote the state reference trajectory, with the corresponding output reference $r(t) = C r_s(t) \in \mathbb{R}$. The state tracking error $e_s(t) \in \mathbb{R}^n$ and output tracking error $e(t) \in \mathbb{R}$ are given by
\begin{equation}
    \label{eq:output_error_def}
    e_s(t) = r_s(t) - x(t), \quad e(t) = r(t) - y(t) = C e_s(t).
\end{equation}

\begin{assum}
\label{assum:reference_feedforward_bounds}
The state reference trajectory $r_s(t) \in \mathbb{R}^n$ and the feedforward control input $u_{ff}(t) \in \mathbb{R}$, along with their first time derivatives, are bounded. There exist known constants $\bar{r}_0, \bar{r}_1, \Delta_{f}, \Delta_{f1} > 0$ such that for almost all $t \geq 0$:
\begin{equation*}
\begin{aligned}
    \|r_s(t)\| &\leq \bar{r}_0, \quad \|\dot{r}_s(t)\| \leq \bar{r}_1, \\
    |u_{ff}(t)| &\leq \Delta_{f}, \quad |\dot{u}_{ff}(t)| \leq \Delta_{f1}.
\end{aligned}
\end{equation*}
\end{assum}

\begin{assum}
\label{assum:perturbation_noise_bounds}
The unmatched disturbance $w(t) \in \mathbb{R}^n$ is bounded; there exists a known constant $\Delta_w \geq 0$ such that $\|w(t)\| \leq \Delta_w$ for all $t \geq 0$.
\end{assum}


\begin{assum}
\label{assum:disturbance_bound}
The matched disturbance $d(t) \in \mathbb{R}$ satisfies the following conditions:
\begin{enumerate}
    \item \textit{Amplitude bound:} There exists a known constant $\Delta > 0$ such that $|d(t)| \leq \Delta$ for all $t \geq 0$.
    \item \textit{Rate bound:} The signal $d(t)$ is locally absolutely continuous, and there exists a known constant $\Delta_1 > 0$ such that $|\dot{d}(t)| \leq \Delta_1$ for almost all $t \geq 0$.
\end{enumerate}
\end{assum}
 
\begin{rem}
In digital implementations governed by a finite sampling period $T_s > 0$, the rates of change of $r_s(t)$ and $d(t)$ are bounded over each sampling interval. This constraint precludes the infinite derivatives associated with ideal continuous-time steps or impulsive disturbances.
\end{rem}
 
\subsection{Reset Control System}
\label{sec:nominal_reset_system} 
The reset controller, with input $e(t)$ and output $u_r(t)$, is described by the state-space equations
\begin{equation}
\label{eq:reset_controller}
\begin{cases}
    \dot{x}_r(t) = A_R x_r(t) + B_R e(t), & \text{if } \zeta(t) \in \mathcal{F}_r, \\
    x_r(t^+)     = A_\rho x_r(t),         & \text{if } \zeta(t) \in \mathcal{J}_r, \\
    u_r(t)       = C_R x_r(t) + D_R e(t),
\end{cases}
\end{equation}
where $x_r(t) \in \mathbb{R}^{n_r}$ is the controller state. The matrices $A_R \in \mathbb{R}^{n_r \times n_r}$, $B_R \in \mathbb{R}^{n_r \times 1}$, $C_R \in \mathbb{R}^{1 \times n_r}$, and $D_R \in \mathbb{R}$ are constant. The reset matrix $A_\rho \in \mathbb{R}^{n_r \times n_r}$ is given by
\begin{equation}
\label{eq:A_rho}
    A_\rho =
    \begin{bmatrix}
        \gamma & 0_{1 \times (n_r-1)} \\
        0_{(n_r-1) \times 1} & I_{n_r-1}
    \end{bmatrix}, \quad \gamma \in (-1,1),
\end{equation}
which indicates that the reset controller used in this study features a single reset state. Let the reset decision vector $\zeta(t) \in \mathbb{R}^{n_r+1}$ be defined as $\zeta(t) = \bigl[x_r^\top(t), e^\top(t)\bigr]^\top$. The flow set $\mathcal{F}_r$ and jump set $\mathcal{J}_r$ are parameterized by a constant symmetric matrix $M = M^\top \in \mathbb{R}^{(n_r+1)\times(n_r+1)}$ such that
\begin{equation}
\label{eq:F_J_set}
\begin{aligned}
    \mathcal{F}_r &= \bigl\{ \zeta \in \mathbb{R}^{n_r+1} \mid \zeta^\top M \zeta \geq 0 \bigr\}, \\
    \mathcal{J}_r &= \bigl\{ \zeta \in \mathbb{R}^{n_r+1} \mid \zeta^\top M \zeta \leq 0 \bigr\}.
\end{aligned}
\end{equation}

\begin{defn}
\label{defn:cl_rc}
Let the state of the closed-loop reset system (with $u_{sm} \equiv 0$) be denoted by $\chi(t) = \bigl[e_s^\top(t),\; x_r^\top(t)\bigr]^\top \in \mathbb{R}^{n+n_r}$. Substituting the control law $u(t) = u_r(t)$ and tracking error $e(t) = C e_s(t)$ into the plant~\eqref{eq:state_space_system} and reset controller~\eqref{eq:reset_controller}, the hybrid dynamics of the closed-loop reset system are given by
\begin{equation}
\label{eq:chi_nominal}
\mathcal{H}_{\chi}:
\begin{cases}
    \dot{\chi}(t) = A_\chi \chi(t) + E_r\bigl(q_r(t) - w(t)\bigr) + B_c d(t),  & \text{ if } \zeta(t) \in \mathcal{F}_r, \\
    \chi(t^+) = A_{\chi,\rho} \chi(t),& \text{ if } \zeta(t) \in \mathcal{J}_r,
\end{cases}
\end{equation}
where the signal $q_r(t) \in \mathbb{R}^n$ is defined as
\begin{equation}
\label{eq:qr_def}
    q_r(t) = \dot{r}_s(t) - Ar_s(t) - Bu_{ff}(t).
\end{equation}
The system matrices are defined as
\begin{equation}
\label{eq:A_chi_def}
\begin{aligned}
    A_\chi &= 
    \begin{bmatrix} 
        A - BD_RC & -BC_R \\ 
        B_RC & A_R 
    \end{bmatrix} \in \mathbb{R}^{(n+n_r)\times(n+n_r)}, \quad
    E_r  = 
    \begin{bmatrix} 
        I_n \\ 
        0_{n_r \times n} 
    \end{bmatrix} \in \mathbb{R}^{(n+n_r)\times n}, \\ 
    B_c &= 
    \begin{bmatrix} 
        -B \\ 
        0_{n_r \times 1} 
    \end{bmatrix} \in \mathbb{R}^{(n+n_r)\times 1},\quad
    A_{\chi,\rho}  = 
    \begin{bmatrix} 
        I_n & 0_{n \times n_r} \\ 
        0_{n_r \times n} & A_\rho 
    \end{bmatrix} \in \mathbb{R}^{(n+n_r)\times(n+n_r)}.
\end{aligned}
\end{equation}
\end{defn}

\begin{rem}
By Assumption~\ref{assum:reference_feedforward_bounds}, $q_r(t)$ in \eqref{eq:qr_def} is bounded such that $\|q_r(t)\| \leq \bar{q}_r$, where $\bar{q}_r = \bar{r}_1 + \|A\|\bar{r}_0 + \|B\|\Delta_f$. From \eqref{eq:qr_def}, provided the pair $(A, B)$ satisfies the feedforward matching condition, $u_{ff}(t)$ is typically designed to enforce $q_r(t) \equiv 0$.
\end{rem}


\begin{assum} 
\label{assum:iss_reset_input}
For the closed-loop reset system~\eqref{eq:chi_nominal}, there exist a symmetric positive-definite matrix $P \in \mathbb{R}^{(n+n_r) \times (n+n_r)}$ and a constant $\lambda_0 > 0$ such that the continuous flow condition
\begin{equation}
\label{eq:tracking_flow_lmi}
    \chi^\top \bigl(A_\chi^\top P + P A_\chi\bigr) \chi \leq -\lambda_0 \chi^\top P \chi, \quad \forall\, \zeta(t) \in \mathcal{F}_r,
\end{equation}
and the discrete jump condition
\begin{equation}
\label{eq:tracking_jump_lmi}
    \chi^\top \bigl(A_{\chi,\rho}^\top P A_{\chi,\rho} - P\bigr) \chi \leq 0, \quad \forall\, \zeta(t) \in \mathcal{J}_r,
\end{equation}
are satisfied.
\end{assum}

Assumption~\ref{assum:iss_reset_input} guarantees the ISS property of the closed-loop reset system~\eqref{eq:chi_nominal}, as formalized in Corollary~\ref{cor:quadratic_iss_reset_sufficient}.

\begin{cor} 
\label{cor:quadratic_iss_reset_sufficient}
Under Assumption \ref{assum:iss_reset_input}, there exist a symmetric positive-definite matrix $P = P^\top > 0$, a decay rate $\lambda=\lambda_0/4 > 0$, and gains $\sigma_r, \sigma_w, \sigma_d \geq 0$ given by
\begin{equation}
\label{eq:iss_constants}
    \sigma_r = \sigma_w = \frac{4}{\lambda_0}\bigl\|E_r^\top P E_r\bigr\|, \quad 
    \sigma_d = \frac{4}{\lambda_0}\bigl\|B_c^\top P B_c\bigr\|,
\end{equation}
such that the Lyapunov function $V_e(\chi) = \chi^\top P \chi$ satisfies the flow condition
\begin{equation}
\label{eq:nominal_iss_flow}
    \dot{V}_e \leq -\lambda V_e + \sigma_r \|q_r\|^2 + \sigma_w \|w\|^2 + \sigma_d \|d\|^2, \quad \forall\, \zeta(t) \in \mathcal{F}_r,
\end{equation}
and the jump condition
\begin{equation}
\label{eq:nominal_iss_jump}
    V_e(\chi^+) \leq V_e(\chi), \quad \forall\, \zeta(t) \in \mathcal{J}_r.
\end{equation}
\end{cor}
\begin{proof}
    Proof is provided in Appendix \ref{appendix:proof_quadratic_iss_reset}.
\end{proof}

\begin{rem}
\label{rem:assum_iss_meaning}
Corollary~\ref{cor:quadratic_iss_reset_sufficient} establishes the ISS \cite{Sontag,Hespanha} of system~\eqref{eq:chi_nominal}. The flow inequality~\eqref{eq:nominal_iss_flow} guarantees that the continuous dynamics of the Lyapunov candidate $V_e(\chi)$ are dissipative, subject to bounded contributions from the feedforward residual $q_r$, the unmatched process perturbation $w$, and the matched disturbance $d$ via their respective gains \cite{SontagWang}. The jump inequality~\eqref{eq:nominal_iss_jump} enforces the non-expansion of the Lyapunov function across all reset instants \cite[Theorem 13.5]{bainov1989systems}. In the unforced and disturbance-free scenario ($r_s \equiv 0$, $u_{ff} \equiv 0$, $w \equiv 0$, and $d \equiv 0$), Corollary~\ref{cor:quadratic_iss_reset_sufficient} guarantees that $V_e(\chi)$ decays at a minimum exponential rate $\lambda$, rendering system~\eqref{eq:chi_nominal} globally exponentially stable.
\end{rem}

\begin{rem}
\label{rem:lmi_verification}
The $\mathcal{H}_\beta$-condition~\cite{beker2004fundamental} provides a sufficient criterion to verify Assumption~\ref{assum:iss_reset_input} when $\gamma=0$. Let the controller state be partitioned into $n_\rho=1$ reset and $n_{\bar\rho}=n_r-n_\rho$ non-reset dimensions. The condition requires the transfer function $H_\beta(s) = C_\beta(sI - A_\chi)^{-1}B_\rho$, where
\begin{equation}
    C_\beta = \begin{bmatrix} \beta C & P_\rho & 0_{n_\rho \times n_{\bar{\rho}}} \end{bmatrix}, \ 
    B_\rho = \begin{bmatrix} 0_{n_\rho \times n} & I_{n_\rho} & 0_{n_\rho \times n_{\bar{\rho}}} \end{bmatrix}^\top\!,
\end{equation}
to be strictly positive real (SPR) for some vector $\beta\in\mathbb{R}^{n_\rho}$ and symmetric matrix $P_\rho>0$.

If $H_\beta(s)$ is SPR, the Kalman--Yakubovich--Popov (KYP) lemma~\cite[Lemma~6.3]{khalil2002nonlinear} guarantees a matrix $P=P^\top>0$ satisfying the condition $A_\chi^\top P+PA_\chi<0$ (yielding a valid $\lambda_0>0$ for~\eqref{eq:tracking_flow_lmi}) and $PB_\rho=C_\beta^\top$. Partitioning the closed-loop state $\chi=(e_s,\, x_{r,\rho},\, x_{r,\bar\rho})$ and defining the matrix $P$ by conformable blocks $P_{ij}$ (for $i,j \in \{e, \rho, \bar\rho\}$), the matrix product $PB_\rho$ evaluates to the second block-column of $P$. Thus, $PB_\rho = C_\beta^\top$ yields $P_{e\rho}=C^\top\beta^\top$, $P_{\rho\rho}=P_\rho$, and $P_{\rho\bar\rho}^\top = 0$. By setting $\beta=0$, we additionally force $P_{e\rho}=0$. Under standard synthesis (where $P_{e\bar\rho}=0$), this yields a fully block-diagonal matrix $P=\operatorname{diag}(P_{ee},\, P_\rho,\, P_{\bar\rho\bar\rho})$, giving $A_{\chi,\rho}^\top PA_{\chi,\rho}-P\le0$ and satisfying~\eqref{eq:tracking_jump_lmi}.

This condition remains valid for all $\gamma\in(-1,1)$. First, $A_\chi$ is $\gamma$-independent, so~\eqref{eq:tracking_flow_lmi} holds unchanged. Second, evaluating the jump matrix with $A_{\chi,\rho}=\operatorname{diag}(I_n,\,\gamma I_{n_\rho},\,I_{n_{\bar\rho}})$ yields
\begin{equation}
    A_{\chi,\rho}^\top P A_{\chi,\rho} - P = \operatorname{diag}\bigl(0_{n \times n},\; (\gamma^2-1)P_\rho,\; 0_{n_{\bar\rho} \times n_{\bar\rho}}\bigr) \le 0,
\end{equation}
which holds because $P_\rho>0$ and $\gamma^2\le1$, preserving~\eqref{eq:tracking_jump_lmi}.
\end{rem}
Remark \ref{rem:lmi_verification} provides a sufficient condition to satisfy Assumption~\ref{assum:iss_reset_input}; alternatively, a valid matrix $P$ can also be computed via techniques such as convex optimization.

\section{Input-to-State Stability Analysis of the RC-ISMC System}
\label{sec:method}
This section formulates the closed-loop RC-ISMC system and employs a Lyapunov analysis to prove its ISS and UUB against references and disturbances.

\subsection{Integral Sliding Mode Control}
\label{subsec: ismc}
Define the saturation function as
\begin{equation}
\label{eq:sat_def}
    \operatorname{sat}_{\delta}(s)=
    \begin{cases}
       \mathrm{sgn}(s), & \text{if } |s|>\delta, \\
        s/\delta, & \text{if } |s|\leq\delta,
    \end{cases}
\end{equation}
where $\delta>0$ is the boundary-layer thickness and $\operatorname{sgn}(s)=1$ for $s>0$, $\operatorname{sgn}(s)=-1$ for $s<0$.

The ISMC law, denoted by $\mathcal{C}_s$, is defined by
\begin{equation}
\label{eq:Cm1}
    \mathcal{C}_s:
    \begin{cases}
        u_{sm}(t) = G^{-1}\rho\,\operatorname{sat}_{\delta}(s(t)), \\[4pt]
        s(t) = C_m e_s(t) + z(t), \\[4pt]
        \dot{z}(t) = -C_m\bigl[\dot{r}_s(t) - Ax(t) - B\bigl(u_r(t) + u_{ff}(t)\bigr)\bigr],
    \end{cases}
\end{equation}
where $C_m \in \mathbb{R}^{1 \times n}$ and $G = C_mB \neq 0$. The gain $\rho$ is selected such that:
\begin{equation}
\label{eq:rho_gain_condition}
    \rho = |G|\Delta + \|C_m\|\Delta_w + \eta, \quad \eta > 0.
\end{equation}

\subsection{Closed-Loop RC-ISMC System}
\label{sec:closed_loop_def}
Lemma \ref{lem:decoupling} establishes how the sliding variable $s(t)$ in \eqref{eq:Cm1} evolves during flows and across resets. 
\begin{lem} 
\label{lem:decoupling}
Consider the closed-loop RC-ISMC system comprising the plant~\eqref{eq:state_space_system}, the reset controller~\eqref{eq:reset_controller}, and the ISMC law~\eqref{eq:Cm1}. Suppose Assumptions~\ref{assum:reference_feedforward_bounds}--\ref{assum:iss_reset_input} hold. The sliding variable $s(t)$ satisfies the flow dynamics
\begin{equation}
\label{eq:sdot_autonomous}
\begin{aligned}
    \dot{s}(t) &= -G\bigl(u_{sm}(t)+d(t)\bigr) - C_m w(t) \\
               &= -G d(t) - C_m w(t) - \rho\operatorname{sat}_{\delta}\bigl(s(t)\bigr), \ \forall\, \zeta(t) \in \mathcal{F}_r,
\end{aligned}
\end{equation}
and is invariant across discrete resets, satisfying the jump condition
\begin{equation}
\label{eq:s_jump}
    s(t^+) = s(t), \quad \forall\, \zeta(t) \in \mathcal{J}_r.
\end{equation}
\end{lem}

\begin{proof}
Proof is provided in Appendix \ref{appendix:proof_lem:decoupling}.
\end{proof}

\begin{rem}
The sliding variable dynamics~\eqref{eq:sdot_autonomous} demonstrate that the ISMC isolates the matched disturbance into the lumped term $u_{sm} + d$ \cite{Utkin1996}, allowing the sliding-mode controller to function as a disturbance observer, as in \cite{ChakraborttyArcak2009}. In the ideal discontinuous limit ($\delta = 0$) and ignoring unmatched perturbations ($w \equiv 0$), enforcing the exact sliding regime ($\dot{s} \equiv 0$) yields the equivalent control $u_{sm} \equiv -d$, in the sense of Filippov's differential inclusions \cite{Utkin1996}. Because the proposed architecture employs a saturation boundary layer ($\delta > 0$) to preclude chattering in \eqref{eq:Cm1}, this exact cancellation is relaxed. The resulting bounded compensation error is quantified in Lemma~\ref{lem:equivalent_control_error}.
\end{rem}


Having established the invariance of the sliding variable across discrete resets, we now present the formulation of the closed-loop system.

Let the state vector of the closed-loop RC-ISMC system be defined as
\begin{equation}
    \xi_e(t) = \bigl[ e_s^\top(t),\; x_r^\top(t),\; s(t)\bigr]^\top \in \mathbb{R}^{n+n_r+1}.
\end{equation}
The reset condition~\eqref{eq:F_J_set} is triggered by the partial state vector $\zeta(t) = \bigl[x_r^\top(t),\; (C e_s(t))^\top\bigr]^\top$. Consequently, the flow and jump sets mapped to the full state space $\mathbb{R}^{n+n_r+1}$ are
\begin{equation}
\label{eq:C_D_int}
\begin{aligned}
    \mathcal{F}_{\mathrm{int}}
    &= \bigl\{\xi_e \in \mathbb{R}^{n+n_r+1}
       \mid \zeta(\xi_e)^\top M\,\zeta(\xi_e) \geq 0\bigr\}, \\
    \mathcal{J}_{\mathrm{int}}
    &= \bigl\{\xi_e \in \mathbb{R}^{n+n_r+1}
       \mid \zeta(\xi_e)^\top M\,\zeta(\xi_e) \leq 0\bigr\},
\end{aligned}
\end{equation}
where $M = M^\top$ is the condition matrix defined in~\eqref{eq:F_J_set}. 
Substituting the control law $u_r = C_R x_r + D_R C e_s$ into the plant dynamics~\eqref{eq:state_space_system} yields the state tracking error dynamics
\begin{equation} 
\label{eq:es_dot_full}
    \dot{e}_s = (A - BD_RC)e_s - BC_R x_r + q_r - w - B(u_{sm} + d).
\end{equation}
Combining \eqref{eq:output_error_def}, \eqref{eq:reset_controller}, \eqref{eq:sdot_autonomous}, and \eqref{eq:es_dot_full}, the continuous flow map of the RC-ISMC system is defined as
\begin{equation}
\label{eq:interconnected_flow_map}
    F_e(\xi_e, q_r, w, d) = \dot{\xi}_e 
    = A_{\xi}\,\xi_e + E_{\xi}\,q_r + W_{\xi}\,w + B_{\xi}\bigl(u_{sm} + d\bigr),
\end{equation}
where the system matrices are defined as
\begin{equation}
\label{eq:hybrid_flow_matrices}
\begin{aligned}
    A_{\xi} &=
    \begin{bmatrix}
        A_\chi  & 0_{(n+n_r) \times 1} \\
        0_{1 \times (n+n_r)} & 0
    \end{bmatrix}, \quad
    E_{\xi} =
    \begin{bmatrix}
        E_r \\ 0_{1 \times n}
    \end{bmatrix}, \quad 
    W_{\xi} =
    \begin{bmatrix}
        -E_r\\ -C_m
    \end{bmatrix}, \quad
    B_{\xi} =
    \begin{bmatrix}
        B_c \\ -G
    \end{bmatrix}.
\end{aligned}
\end{equation}
By Lemma~\ref{lem:decoupling} and its proof in Appendix \ref{appendix:proof_lem:decoupling}, we have
\begin{equation}
\label{eq:jump_components}
\begin{aligned}
    e_s(t_k^+) &=  e_s(t_k),  \ 
    x_r(t_k^+)  = A_\rho x_r(t_k),  \ 
    s(t_k^+)  = s(t_k). 
\end{aligned}
\end{equation}
From \eqref{eq:jump_components}, the jump map of the RC-ISMC system is given by
\begin{equation}
\label{eq:interconnected_jump_map}
    G_{\mathrm{e}}(\xi_e) = \xi_e^+ = A_{\xi,\rho}\,\xi_e,\quad A_{\xi,\rho} = \operatorname{diag}(A_{\chi,\rho},\; 1).
\end{equation}
Finally, combining \eqref{eq:interconnected_flow_map} and \eqref{eq:interconnected_jump_map}, the closed-loop RC-ISMC system is formulated as
\begin{equation}
\label{eq:hybrid_interconnected}
    \mathcal{H}_{\mathrm{int}}:
    \begin{cases}
        \dot{\xi}_e = F_e(\xi_e, q_r, w, d), & \text{if } \xi_e \in \mathcal{F}_{\mathrm{int}}, \\
        \xi_e^{+} = G_{\mathrm{e}}(\xi_e), & \text{if } \xi_e \in \mathcal{J}_{\mathrm{int}}.
    \end{cases}
\end{equation}
The block diagram of this RC-ISMC system is depicted in Fig.~\ref{fig:closed_loop_reset_ismc_system}.

\begin{figure}[htp]
	\centering
	\includegraphics[width=0.6\columnwidth]{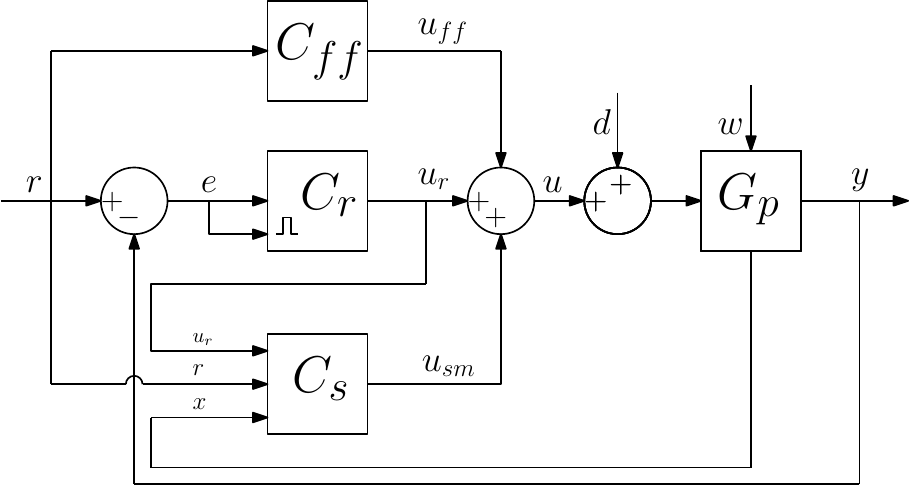}
	\caption{Block diagram of the closed-loop RC-ISMC system.}
	\label{fig:closed_loop_reset_ismc_system}
\end{figure}

\begin{assum} 
\label{assum:zeno}
The system $\mathcal{H}_{\mathrm{int}}$ does not exhibit Zeno behavior. Specifically, there exists a minimum dwell time $\tau_{\min} > 0$ such that any two consecutive reset instants $t_k$ and $t_{k+1}$ satisfy $t_{k+1} - t_k \geq \tau_{\min}$ for all $k \in \mathbb{N}$.
\end{assum}

\begin{rem}
In practical digital implementations, Assumption~\ref{assum:zeno} is inherently satisfied. The finite sampling period of the hardware and the use of a zero-order hold (ZOH) mechanism physically enforce a positive minimum dwell time between consecutive jumps. A formal treatment of these discretization effects is detailed in~\cite{Banos2016}.
\end{rem}
\subsection{Reachability of the Sliding Surface}
\label{sec:reachability}
\begin{lem}
\label{lem:finite_time_s}
Consider the closed-loop RC-ISMC system $\mathcal{H}_{\mathrm{int}}$~\eqref{eq:hybrid_interconnected} under Assumptions~\ref{assum:reference_feedforward_bounds}--\ref{assum:zeno}. The boundary layer $\Omega_\delta = \{s \in \mathbb{R} \mid |s| \leq \delta\}$ is positively invariant, and any state trajectory with an initial condition $s(0) \notin \Omega_\delta$ reaches $\Omega_\delta$ in a finite time $t^*$ bounded by
\begin{equation}
\label{eq:reaching_time}
    t^* \leq ({|s(0)| - \delta})/{\eta}, \quad \eta = \rho - |G|\Delta - \|C_m\|\Delta_w > 0.
\end{equation}
Furthermore, initializing the integral state as $z(0) = -C_m e_s(0)$ enforces $s(0) = 0 \in \Omega_\delta$, thereby eliminating the reaching phase entirely ($t^* = 0$).
\end{lem}

\begin{proof}
Proof is provided in Appendix \ref{appendix:proof_lem:finite_time_s}.
\end{proof}

\subsection{Lyapunov Stability Analysis}

This section presents the stability analysis for the closed-loop RC-ISMC system. The analysis begins with Lemma~\ref{lem:equivalent_control_error}, which bounds the residual disturbance error $v = u_{sm} + d$ within the boundary layer.

\begin{lem} 
\label{lem:equivalent_control_error}
Consider the closed-loop RC-ISMC system $\mathcal{H}_{\mathrm{int}}$~\eqref{eq:hybrid_interconnected} under Assumptions~\ref{assum:reference_feedforward_bounds}--\ref{assum:zeno}. Suppose the system state reaches the positively invariant boundary layer $\Omega_\delta = \{s \in \mathbb{R} \mid |s| \leq \delta\}$ at some finite time $t^* \geq 0$. Then, for all $t > t^*$, the signal $v(t) = u_{sm}(t) + d(t)$ is differentiable and satisfies the differential equation
\begin{equation}
\label{eq:v_differential}
    \dot{v}(t) +  ({\rho}/{\delta})v(t) = \dot{d}(t) - G^{-1}({\rho}/{\delta})C_m w(t).
\end{equation}
From \eqref{eq:v_differential}, as $t \to \infty$, the magnitude of $v(t)$ satisfies the ultimate bound
\begin{equation}
\label{eq:v_ultimate_bound}
    \limsup_{t \to \infty} |v(t)| \leq \bar{v}_\delta = \frac{\delta \Delta_1}{\rho} + \frac{\|C_m\|\Delta_w}{|G|}.
\end{equation}
\end{lem}

\begin{proof}
    Proof is provided in Appendix \ref{appendix:proof_lem equivalent_control_error}.
\end{proof}

\begin{thm} 
\label{thm:continuous_flow_stability}
Consider the closed-loop RC-ISMC system $\mathcal{H}_{\mathrm{int}}$~\eqref{eq:hybrid_interconnected} under Assumptions~\ref{assum:reference_feedforward_bounds}--\ref{assum:zeno}. Let $P = P^\top > 0$ and $\lambda_0 > 0$ be the matrix and scalar satisfying Assumption~\ref{assum:iss_reset_input}. For the system state vector $\xi_e = [\chi^\top,\; s]^\top$, define the ISS-Lyapunov function 
\begin{equation}
\label{eq:composite_lyapunov}
    V(\xi_e) = \chi^\top P \chi +  {s^2}/{2}.
\end{equation}
Then, the following properties hold:
\begin{enumerate}[label=(\roman*)]
\item \textit{Continuous flow dissipation:} For any arbitrary $\epsilon > 0$, there exists a finite time $T_\epsilon \geq t^*$ such that the derivative of $V$ in \eqref{eq:composite_lyapunov} along the continuous flows $\xi_e \in \mathcal{F}_{\mathrm{int}}$ for all $t \geq T_\epsilon$ satisfies
\begin{equation}
\label{eq:global_flow_bound}
    \dot{V}(\xi_e) \leq -\alpha V(\xi_e) + c_q + c_w + c_{\mathrm{coup},\epsilon} + c_s,
\end{equation}
where the decay rate is $\alpha = \min\bigl(\frac{\lambda_0}{4},\, \frac{\rho}{\delta}\bigr)$, and 
\begin{equation}
\label{eq:constants_def}
\begin{aligned}
    &c_q  = \frac{4}{\lambda_0}\bigl\|E_r^\top P E_r\bigr\|\bar{q}_r^2, \ c_{\mathrm{coup},\epsilon} = \frac{4}{\lambda_0}\bigl\|B_c^\top P B_c\bigr\|(\bar{v}_\delta + \epsilon)^2, \\
    &  c_w = \frac{4}{\lambda_0}\bigl\|E_r^\top P E_r\bigr\|\Delta_w^2,\  c_s = \frac{\delta}{\rho}\bigl(G^2\Delta^2 + \|C_m\|^2\Delta_w^2\bigr),
\end{aligned}
\end{equation}
with $\bar{v}_\delta$ defined in~\eqref{eq:v_ultimate_bound}.

\item \textit{Jump non-expansion:} Across discrete resets $\xi_e \in \mathcal{J}_{\mathrm{int}}$, the Lyapunov function satisfies
\begin{equation}
\label{eq:jump_nonexpansive}
    V(\xi_e^+) \leq V(\xi_e).
\end{equation}
\end{enumerate}
Thus, the system is ISS with respect to $q_r(t)$, $w(t)$, and $d(t)$. Furthermore, the state $\xi_e(t)$ is UUB with
\begin{equation}
\label{eq:ultimate_bound}
    \limsup_{t\to\infty}\|\xi_e(t)\| \leq \sqrt{\frac{c_{\mathrm{total}} }{\alpha\,\lambda^\ast}},
\end{equation}
where $c_{\mathrm{total}} = c_q + c_w + c_{\mathrm{coup},\delta} + c_s$ evaluated precisely at $c_{\mathrm{coup},\delta} = \frac{4}{\lambda_0}\bigl\|B_c^\top P B_c\bigr\|\bar{v}_\delta^2$, and $\lambda^\ast = \min\bigl(\lambda_{\min}(P), 1/2\bigr)$.
\end{thm}

\begin{proof}
    Proof is provided in Appendix \ref{appendix:proof_lemma_flow}.
\end{proof}


\begin{rem} 
\label{rem:ultimate_bound}
As established in the proof of Theorem~\ref{thm:continuous_flow_stability} (Appendix~\ref{appendix:proof_lemma_flow}, \eqref{eq:v_chi}), the state tracking error $e_s(t)$ is a sub-vector of $\chi(t)$. Applying the spectral lower bound $V_e(\chi) \geq \lambda_{\min}(P)\|\chi\|^2$ thus yields the conservative limit:
\begin{equation}
\label{eq:es_ultimate_bound}
    \limsup_{t \to \infty} \|e_s(t)\| \leq  \limsup_{t\to\infty}\|\chi(t)\| \leq \sqrt{\frac{c_q + c_w + c_{\mathrm{coup},\delta}}{\alpha \lambda_{\min}(P)}}.
\end{equation}
\end{rem}


\begin{rem} 
\label{rem:regulation_special_case}
Consider the closed-loop system $\mathcal{H}_{\mathrm{int}}$~\eqref{eq:hybrid_interconnected} under Assumptions~\ref{assum:reference_feedforward_bounds}--\ref{assum:zeno} in the unperturbed case, where $d \equiv 0$, $w \equiv 0$, and $q_r \equiv 0$. The condition $q_r \equiv 0$ implies either a pure regulation problem where $r_s \equiv 0$ and $u_{ff} \equiv 0$, or that the feedforward control $u_{ff}$ perfectly tracks the reference dynamics. In this scenario, we have $c_q = c_w = c_{\mathrm{coup},\delta} = c_s = 0$. Consequently, the state tracking error converges asymptotically to zero:
\begin{equation}
    \lim_{t\to\infty} \|e_s(t)\| = 0.
\end{equation}
By evaluating equations \eqref{eq:global_flow_bound} and~\eqref{eq:jump_nonexpansive} under these zero-input conditions, the origin of the error dynamics recovers global asymptotic stability (GAS), consistent with \cite[Theorem~13.5]{bainov1989systems}.
\end{rem}

\begin{rem} 
\label{rem:limiting_behavior_delta}
As the boundary layer thickness vanishes ($\delta \to 0^+$), the ISMC law~\eqref{eq:Cm1} recovers the ideal signum switching function, and the boundary layer collapses to the exact sliding surface $\mathcal{S} = \{s = 0\}$. Consequently, the residual sliding energy vanishes ($c_s \to 0$) and the decay rate saturates at $\alpha = \lambda_0/4$. On $\mathcal{S}$, the equivalent control perfectly rejects the matched disturbance $d(t)$ (yielding $u_{sm} + d = 0$), consistent with the classical sliding mode theory~\cite{Utkin1996}. Evaluating the state tracking error bound~\eqref{eq:es_ultimate_bound} under these limits yields
\begin{equation}
\label{eq:ideal_limit_bound}
    \limsup_{t\to\infty}\|e_s(t)\| \leq 2 \sqrt{\frac{c_q + c_w + c_{\mathrm{coup},0}}{\lambda_0 \lambda_{\min}(P)}}.
\end{equation}
where $c_q$ and $c_w$ remain as defined in~\eqref{eq:constants_def}, and the coupling constant reduces to $c_{\mathrm{coup},0} \triangleq \frac{4}{\lambda_0}\|B_c^\top P B_c\| \bigl( \frac{\|C_m\|\Delta_w}{|G|} \bigr)^2$. While the proposed ISMC is designed to attenuate matched disturbances, future work could suppress the unmatched residual terms $c_w$ and $c_{\mathrm{coup},0}$ by incorporating techniques such as the sliding surface modification method in \cite{castanos2006analysis}.
\end{rem}

\section{Case Study: Numerical Validation}
\label{sec:results}

This section presents a numerical simulation of a second-order mass-spring-damper system to illustrate the UUB property of the RC-ISMC architecture.

\subsection{System Design and Exogenous Inputs}
The linear plant ($m = 1.0$\,kg, $c = 1.0$\,N$\cdot$s/m, $k = 10.0$\,N/m) is defined by~\eqref{eq:state_space_system} with the system matrices
\begin{equation}
    A = \begin{bmatrix} 0 & 1 \\ -10 & -1 \end{bmatrix}, \quad
    B = \begin{bmatrix} 0 \\ 1 \end{bmatrix}, \quad
    C = \begin{bmatrix} 1 & 0 \end{bmatrix},
\end{equation}
corresponding to the transfer function $P(s) = 1/(s^2+s+10)$. 

A First-Order Reset Element (FORE)~\cite{5712180} in parallel with a proportional feedthrough term serves as the nominal controller. It is defined by~\eqref{eq:reset_controller} with $A_R = -2.0$, $B_R = 1.0$, $C_R = 15.0$, and $D_R = 2.0$, yielding the base linear controller $C(s) = (2s+19)/(s+2)$. Controller resets to the origin ($\gamma = 0$) are triggered when the tracking error crosses zero ($e(t) = 0$), corresponding to the matrix $M = \operatorname{diag}(0, 1)$ in \eqref{eq:F_J_set}, with a minimum dwell time of $0.05\,$s.

The state tracks a reference trajectory $r_s(t) = [\,0.5\sin(1.5t),\ 0.75\cos(1.5t)\,]^\top$. Since $\dot r_{s,1}(t) \equiv r_{s,2}(t)$ by construction, $r_s(t)$ is already matched through $B$; hence no unmatched feedforward channel is needed, and a scalar analytic feedforward $u_{ff}(t)$ suffices to cancel the residual mismatch $\dot r_s - Ar_s$ in the second state equation. The initial condition is chosen outside the sliding boundary layer, $x_0=[1.2,\,-0.8]^\top$, $x_{R,0}=0.15$, giving $s(0)=0.80 \gg \delta$, so that the reaching-phase transient is visible in simulation.

The system is subjected to an oscillatory matched disturbance $d(t) = 3.0\sin(4.0t) + 2.0\cos(10.0t)$, bounded by $\Delta = 5$, and the unmatched perturbation $w(t) = 0.5\sin(5t)$, bounded by $\Delta_w = 0.5$. The sliding surface matrix is designed as $C_m = [2.0,\, 1.0]$, yielding $G = C_m B = 1.0$. The switching gain $\rho = 6.5$ strictly satisfies the reachability condition $\rho > |G|\Delta+ \|C_m\|\Delta_w$, with margin $\eta = \rho-|G|\Delta-\|C_m\|\Delta_w \approx 0.38$. Finally, a boundary layer thickness of $\delta = 0.01$ is selected to eliminate high-frequency chattering.

\subsection{Numerical Verification of ISS for the Reset System}

With the closed-loop state vector defined as $\chi = \begin{bmatrix} e_{s1} & e_{s2} & x_r \end{bmatrix}^\top$, where the tracking error is $e = e_{s1}$, the continuous-flow matrix of the reset system ($u_{sm} \equiv 0$) is given by
\begin{equation}
    A_\chi = \begin{bmatrix} A - BD_RC & -BC_R \\ B_RC & A_R \end{bmatrix}
    = \begin{bmatrix} 0 & 1 & 0 \\ -12 & -1 & -15 \\ 1 & 0 & -2 \end{bmatrix}.
\end{equation}
Thus, $A_\chi$ is strictly Hurwitz and decoupled from the ISMC parameters $(\rho,\delta,\Delta)$.

To establish the required Lyapunov matrix, equations \eqref{eq:tracking_flow_lmi} and \eqref{eq:tracking_jump_lmi} from Assumption~\ref{assum:iss_reset_input} are solved numerically, subject to the constraint $P_{23} = P_{32} = 0$. This yields the feasible solution
\begin{equation}
    P = \begin{bmatrix}
        1.6787 & 0.0922 & 0.9217 \\
        0.0922 & 0.1260 & 0 \\
        0.9217 & 0 & 4.7427
    \end{bmatrix}, \qquad \lambda_0 = 0.1314.
\end{equation}
For this specific matrix $P$, the maximum eigenvalue of the continuous-flow condition matrix, $\lambda_{\max}(A_\chi^\top P + P A_\chi + \lambda_0 P)$, evaluates to $-2.13 \times 10^{-8} < 0$.
Furthermore, when evaluated on the jump set $\mathcal{J}_r$ (where $\gamma = 0$ and $e_{s1} = 0$), the quadratic form for the discrete transition reduces to:
\begin{equation}
    \chi^\top \bigl(A_{\chi,\rho}^\top P A_{\chi,\rho} - P\bigr) \chi \Big|_{e_{s1}=0} = -4.7427\,x_r(t_k)^2 \leq 0,
\end{equation}
which confirms that the jump condition in \eqref{eq:tracking_jump_lmi} is satisfied, guaranteeing \eqref{eq:nominal_iss_jump}.

To complete the verification, we apply Remark~\ref{rem:lmi_verification} to evaluate the $\mathcal{H}_\beta$-condition by establishing the SPR of $H_\beta(s) = C_{h_\beta}(sI - A_\chi)^{-1}B_{h_\beta}$. Utilizing $C_{h_\beta} = \begin{bmatrix} 0 & 0 & P_{33} \end{bmatrix}$, numerical evaluation yields $\operatorname{Re}\{H_\beta(j\omega)\} \geq 9.49 \times 10^{-6} > 0$ over $\omega \in [10^{-2}, 10^3]$\,rad/s. Consequently, Assumption~\ref{assum:iss_reset_input} is fulfilled, and the ISS of the closed-loop reset system is guaranteed via Corollary~\ref{cor:quadratic_iss_reset_sufficient}.

\subsection{Simulation Results}
\subsubsection{Lyapunov Stability and Sliding Dynamics}
Figure~\ref{fig:flow_jump_s} illustrates the state trajectories and the sliding variable dynamics of the RC-ISMC system. The top and middle subplots show the Lyapunov function $V(\xi_e)$ decaying toward its ultimate bound during flows and satisfying the non-positive jump condition $V(t_k^+) \leq V(t_k)$ at resets, thereby verifying Theorem~\ref{thm:continuous_flow_stability}. The bottom subplot demonstrates that the ISMC law $u_{sm}$ drives $s(t)$ from $s(0)=0.80$ into the boundary layer ($|s| \leq \delta = 0.01$), remaining confined thereafter.
\begin{figure}[htbp]
    \centering
    \includegraphics[width=0.9\columnwidth]{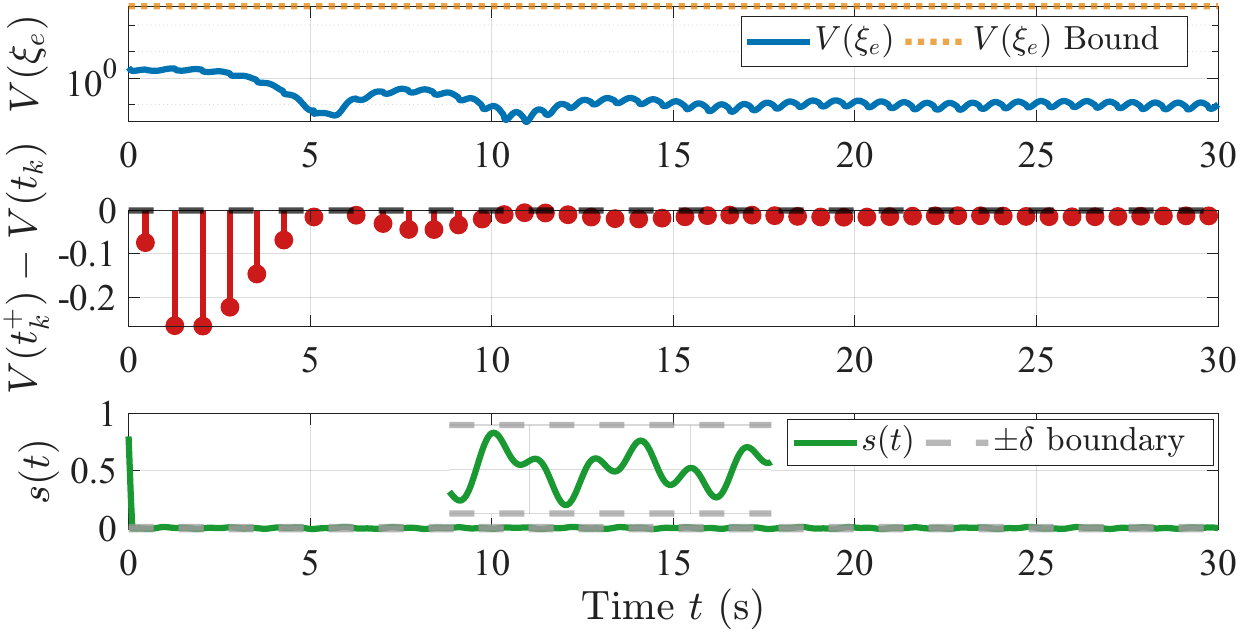}
    \caption{RC-ISMC system response. (Top) Evolution of the composite Lyapunov function $V(\xi_e)$ during continuous flows. (Middle) Non-positive increment $V(t_k^+) - V(t_k) \leq 0$ at reset instants. (Bottom) Trajectory of the sliding variable $s(t)$ starting from $s(0)=0.80$ and converging into the boundary layer $|s| \leq \delta$.}
    \label{fig:flow_jump_s}
\end{figure}

\subsubsection{Ultimate Boundedness of the State Tracking Error}
Since the feedforward term $u_{ff}(t)$ is designed to exactly cancel the reference mismatch, we have $q_r(t) \equiv 0$ and thus $c_q = 0$. Evaluating the remaining UUB constants in \eqref{eq:constants_def} yields $c_w = 12.8149$ and $c_{\mathrm{coup},\delta} = 5.1988$. The associated decay rate is $\alpha = \min\bigl(\lambda_0/4,\ \rho/\delta\bigr) = \lambda_0/4 \approx 0.0329$. Projecting the composite Lyapunov ellipsoid onto the tracking-error subspace then yields the ultimate bound
\begin{equation}
    \limsup_{t\to\infty}\|e_s(t)\| \;\leq\;   67.6362.
\end{equation} 

Figure~\ref{fig:es_uub} compares this conservative theoretical limit against the simulated tracking error $\|e_s(t)\|$ for both the standalone reset and RC-ISMC systems. This inherent conservatism arises primarily from the decoupled Lyapunov structure $V(\xi_e) = V_e(\chi) + \frac{1}{2}s^2$, sequential applications of Young's inequality, and the minimum-eigenvalue projection $\lambda_{\min}(P)$. While such a loose theoretical bound limits its direct utility for precise practical tuning, the simulation reveals that the RC-ISMC system achieves a lower steady-state tracking error than the baseline reset controller. This demonstrates the architecture's practical potential for high-precision tracking applications, motivating comprehensive parameter optimization as a key direction for future research.

\begin{figure}[htbp]
    \centering
    \includegraphics[width=0.9\columnwidth]{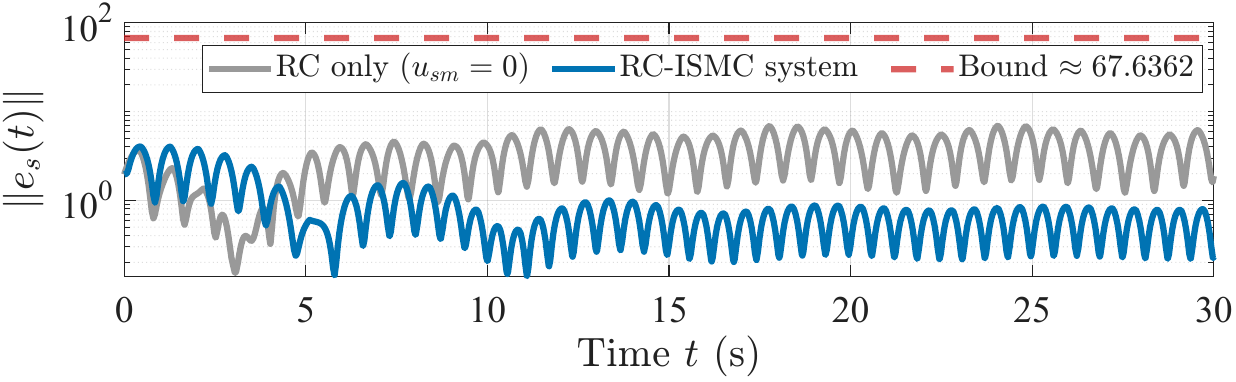}
   \caption{State tracking error $\|e_s(t)\|$ of the RC-ISMC system compared against the standalone reset controller and the theoretical UUB bound.}
    \label{fig:es_uub}
\end{figure}

Furthermore, a numerical illustration of the unperturbed regulation case discussed in Remark~\ref{rem:regulation_special_case} is presented in Fig.~\ref{fig:es_regulation}. By the end of the $600$\,s simulation, the state norm decays to $\|e_s\| = 1.001 \times 10^{-16}$, reaching the limit of numerical precision. This behavior aligns with the theoretically predicted GAS.

\begin{figure}[htbp]
    \centering
    \includegraphics[width=0.9\columnwidth]{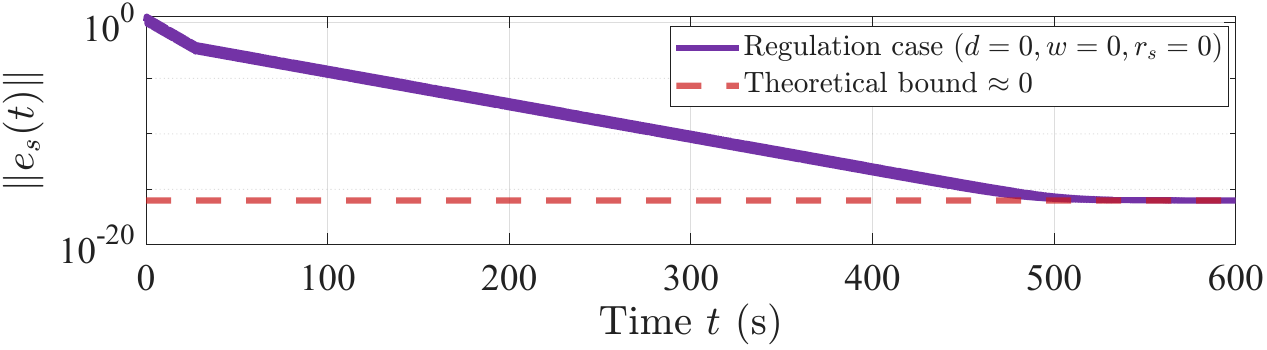}
    \caption{Unperturbed regulation case ($d=0$, $w=0$, $r_s=0$): The tracking error norm $\|e_s(t)\|$ converges asymptotically to zero.}
    \label{fig:es_regulation}
\end{figure}

\section{Conclusion}
\label{sec:conc}

This study presented an ISS analysis of the RC-ISMC system. First, the finite-time reachability of a boundary layer around the sliding surface was established. A Lyapunov-based analysis then proved that the RC-ISMC system is ISS and UUB under exogenous references and both matched and unmatched disturbances, recovering GAS in the absence of perturbations. Numerical results illustrated the theoretical findings.

Several limitations of this study point to directions for future work. The present results establish stability guarantees only; the derived bounds are conservative, owing to the reliance on Young's inequality and suboptimal Lyapunov certificate, and translating these guarantees into practical performance improvements will require dedicated tuning methods for the design parameters (e.g., $\rho$, $\delta$, $\gamma$). The analysis is developed under full-state feedback and focuses on a reset element with a single reset state, albeit one widely used in practice, while the ISMC layer is designed to attenuate matched disturbances only. Future work will therefore pursue output-feedback designs, reset controllers with higher-order reset states, rejection of a broader class of perturbations, tighter bounds, systematic tuning methodologies, and experimental validation.
 
\bibliographystyle{unsrt}
\bibliography{Ref}

\appendix
\subsection{Proof of Corollary \ref{cor:quadratic_iss_reset_sufficient}}
\label{appendix:proof_quadratic_iss_reset}
\begin{proof}
\textit{Step 1: Lyapunov derivative during flows.}
Let $V_e = \chi^\top P\chi$. Differentiating along the augmented flow $\dot{\chi} = A_\chi\chi + E_r(q_r - w) + B_c d$ in \eqref{eq:chi_nominal} gives
\begin{equation}
\label{eq: dotve1}
\begin{aligned}
    \dot{V}_e
    &= \dot{\chi}^\top P\chi + \chi^\top P\dot{\chi} = \chi^\top(A_\chi^\top P + PA_\chi)\chi +   2\chi^\top P E_r q_r - 2\chi^\top P E_r w + 2\chi^\top P B_c d.
\end{aligned}
\end{equation}
By \eqref{eq:tracking_flow_lmi}, we have
\begin{equation}
\label{eq: chi_leq1}
    \chi^\top(A_\chi^\top P + PA_\chi)\chi \leq -\lambda_0\,\chi^\top P\chi = -\lambda_0 V_e.
\end{equation}
Combining \eqref{eq: dotve1} and \eqref{eq: chi_leq1}, we have
\begin{equation}
\label{eq:Ve_dot_before_young}
    \dot{V}_e \leq -\lambda_0 V_e + 2\chi^\top P E_r q_r - 2\chi^\top P E_r w + 2\chi^\top P B_c d.
\end{equation}
For any vectors $a, b$ and scalar $\mu > 0$, Young's inequality states $2a^\top b \leq \mu\|a\|^2 + \mu^{-1}\|b\|^2$. Applying this with $\mu = \lambda_0/4$ to the three cross terms, we obtain
\begin{equation}
\label{eq:young_qr}
\begin{aligned}
    2\chi^\top P E_r q_r
    &= 2(P^{1/2}\chi)^\top(P^{1/2}E_r q_r) \\
    &\leq \frac{\lambda_0}{4}\|P^{1/2}\chi\|^2
      + \frac{4}{\lambda_0}\|P^{1/2}E_r q_r\|^2 \\
    &\leq \frac{\lambda_0}{4} V_e
      + \frac{4}{\lambda_0}\bigl\|E_r^\top P E_r\bigr\|\,\|q_r\|^2,
\end{aligned}
\end{equation}
Similar to \eqref{eq:young_qr}, we have
\begin{equation}
\label{eq:young_w}
\begin{aligned}
    -2\chi^\top P E_r w 
    &\leq \frac{\lambda_0}{4} V_e  + \frac{4}{\lambda_0}\bigl\|E_r^\top P E_r\bigr\|\,\|w\|^2,\\
    2\chi^\top P B_c d
    &\leq \frac{\lambda_0}{4} V_e + \frac{4}{\lambda_0}\bigl\|B_c^\top P B_c\bigr\|\,\|d\|^2.
\end{aligned}
\end{equation}
Substituting~\eqref{eq:young_qr} and \eqref{eq:young_w} into~\eqref{eq:Ve_dot_before_young}, we have
\begin{equation}
\begin{aligned}
    \dot{V}_e &\leq -\frac{\lambda_0}{4} V_e
      + \frac{4}{\lambda_0}\bigl\|E_r^\top P E_r\bigr\|\,\|q_r\|^2 + \frac{4}{\lambda_0}\bigl\|E_r^\top P E_r\bigr\|\,\|w\|^2
     + \frac{4}{\lambda_0}\bigl\|B_c^\top P B_c\bigr\|\,\|d\|^2,
\end{aligned}
\end{equation}
establishing~\eqref{eq:nominal_iss_flow} with $\lambda = \lambda_0/4$, $\sigma_r = \sigma_w = 4\|E_r^\top P E_r\|/\lambda_0$, and $\sigma_d = 4\|B_c^\top P B_c\|/\lambda_0$ as defined in~\eqref{eq:iss_constants}.

\textit{Step 2: Non-expansion at reset instants.}
At a reset instant $t_k$, the state undergoes the jump $\chi^+ = A_{\chi,\rho}\,\chi$ in \eqref{eq:chi_nominal}. Therefore,
\begin{equation}
    V_e(\chi^+)
    = (\chi^+)^\top P\,\chi^+
    = \chi^\top A_{\chi,\rho}^\top P A_{\chi,\rho}\,\chi.
\end{equation}
From \eqref{eq:tracking_jump_lmi}, we obtain $ \chi^\top A_{\chi,\rho}^\top P A_{\chi,\rho}\,\chi \leq \chi^\top P\chi$, 
and therefore $V_e(\chi^+) \leq V_e(\chi)$, establishing~\eqref{eq:nominal_iss_jump}.
\end{proof}

\subsection{Proof of Lemma \ref{lem:decoupling}}
\label{appendix:proof_lem:decoupling}

\begin{proof}
    \textit{Step 1: Continuous-time derivative of the sliding variable.}
Let $t_k$ denote the reset instants satisfying $\xi_e(t_k) \in \mathcal{J}_{\mathrm{int}}$. During the continuous-flow intervals $t \in (t_k, t_{k+1})$, differentiating the sliding variable $s(t)$ in~\eqref{eq:Cm1} along the system dynamics~\eqref{eq:state_space_system} yields the sliding dynamics:
\begin{equation}
\label{eq:s_dot}
\begin{aligned}
    \dot{s}(t) 
    = -G \bigl(u_{sm}(t) + d(t)\bigr) - C_m w(t).
\end{aligned}
\end{equation}
Substituting $u_{sm}$ from \eqref{eq:Cm1} into~\eqref{eq:s_dot} establishes \eqref{eq:sdot_autonomous}.

\textit{Step 2: Continuity of the plant state $x(t)$ across resets.}
From \eqref{eq:output_error_def} and \eqref{eq:reset_controller}, we have
\begin{equation}
\label{eq:ur_rs_x}
u_r(t) = C_Rx_r+D_RCr_s-D_RCx(t).    
\end{equation}
From \eqref{eq:state_space_system} and \eqref{eq:ur_rs_x}, we obtain
\begin{equation}
\label{eq:dotx12}
    \dot{x}(t) = A_c\, x(t) + B C_R\, x_r(t) + \psi(t),
\end{equation}
where $A_c = A-BD_RC$ and $\psi(t) = B(u_{sm}(t) + d(t) + u_{ff}(t)) + B D_R C\, r_s(t) + w(t)$.  
From \eqref{eq:dotx12}, within the interval $[t_k,\, t_k + h]$, we have
\begin{equation}
\label{eq:x_transition_expanded}
\begin{aligned}
    x(t_k + h) &= e^{A_c h}\, x(t_k) + \int_{t_k}^{t_k+h} e^{A_c(t_k+h-\tau)}\, \psi(\tau)\,\mathrm{d}\tau  + B C_R \int_{t_k}^{t_k+h} e^{A_c(t_k+h-\tau)}\, x_r(\tau)\,\mathrm{d}\tau.
\end{aligned}
\end{equation}
From \eqref{eq:sat_def} and \eqref{eq:Cm1}, we have
\begin{equation}
\label{eq:usm_bound}
    \|u_{sm}(t)\| \leq \rho \|G^{-1}\|.
\end{equation}
By Assumptions~\ref{assum:reference_feedforward_bounds}--\ref{assum:disturbance_bound} and \eqref{eq:usm_bound}, the components of $\psi$ are bounded ($\|d\| \leq \Delta$, $\|w\| \leq \Delta_w$, $\|u_{ff}\| \leq \Delta_f$, $\|r_s\| \leq \bar{r}_0$, $\|u_{sm}\| \leq \rho \|G^{-1}\|$). Because the reset matrix $A_\rho$ in \eqref{eq:A_rho} acts as a bounded linear operator, the post-jump state $x_r(t_k^+)$ is strictly finite. Furthermore, since the continuous-time dynamics preclude finite-time escape, the system trajectories cannot grow unbounded over any finite inter-jump interval of length $h$. Therefore, there exist finite local constants $M_\psi, M_{x_r} > 0$ such that $\sup_{\tau \in (t_k,\, t_k+h]} \|\psi(\tau)\| \leq M_\psi$ and $\sup_{\tau \in (t_k,\, t_k+h]} \|x_r(\tau)\| \leq M_{x_r}$. Additionally, by the continuity of the matrix exponential over a compact interval, we define the strictly finite constant $M_e = \sup_{\tau \in [t_k,\, t_k+h]} \|e^{A_c(t_k+h-\tau)}\|$. Taking the norm of the integral terms in~\eqref{eq:x_transition_expanded} yields the strict upper bounds:
 \begin{equation}
\label{eq:integral_bounds_psi}
\begin{aligned}
    \left\|\int_{t_k}^{t_k+h} e^{A_c(t_k+h-\tau)}\, \psi(\tau)\,\mathrm{d}\tau\right\| 
    &\leq 
    h M_e M_\psi, \\ 
    \left\|B C_R \int_{t_k}^{t_k+h} e^{A_c(t_k+h-\tau)}\, x_r(\tau)\,\mathrm{d}\tau\right\| 
    &\leq h \|B C_R\|\, M_e M_{x_r}.
\end{aligned}
\end{equation}
Taking the limit as $h \to 0^+$ in~\eqref{eq:integral_bounds_psi} yields:
\begin{equation}
    \lim_{h \to 0^+} (h M_e M_\psi) = 0, \quad \lim_{h \to 0^+} (h \|B C_R\|\, M_e M_{x_r}) = 0.
\end{equation}
Thus, evaluating the limit $h \to 0^+$ on~\eqref{eq:x_transition_expanded} yields
\begin{equation}
\label{eq:x_jump_continuity}
    x(t_k^+) = \lim_{h \to 0^+} x(t_k + h) = e^0 x(t_k) + 0 + 0 = x(t_k).
\end{equation}

\textit{Step 3: Continuity of the state tracking error $e_s(t)$.}
At any reset instant $t_k$, we have $r_s(t_k^+) = r_s(t_k)$. Substituting this property and \eqref{eq:x_jump_continuity} into \eqref{eq:output_error_def}, we obtain
\begin{equation}
\label{eq:es_jump}
    e_s(t_k^+) = r_s(t_k^+) - x(t_k^+) = r_s(t_k) - x(t_k) = e_s(t_k).
\end{equation}

    \textit{Step 4: Preservation of the integral state and sliding variable.}
    From~\eqref{eq:Cm1}, the increment of the integral state $z$ over the interval $(t_k,\, t_k + h]$ is:
\begin{equation}
\label{eq:z_jump_integral}
\begin{aligned}
    z(t_k^+) - z(t_k) = \lim_{h \to 0^+} \int_{t_k}^{t_k+h} -C_m \bigl[\dot{r}_s(\tau) - A x(\tau) - B u_r(\tau)- B u_{ff}(\tau)\bigr]\mathrm{d}\tau.
\end{aligned}
\end{equation}
From \eqref{eq:reset_controller}, we have
\begin{equation}
\label{eq:ur_jump_eval}
    u_r(t_k^+) - u_r(t_k) = C_R(A_\rho - I)x_r(t_k).
\end{equation}
Because the continuous-time flow is governed by linear dynamics in \eqref{eq:reset_controller}, the reset controller state precludes finite-time escape. Consequently, at any finite reset instant $t_k < \infty$, the pre-jump state $x_r(t_k)$ is guaranteed to be finite, ensuring that \eqref{eq:ur_jump_eval} is well-defined and bounded.

Since $\dot{r}_s$, $x$, and $u_{ff}$ are bounded, and the jump in $u_r$ is finite per~\eqref{eq:ur_jump_eval}, there exist finite constants $M_{\dot{r}},\, M_x,\, M_u,\, M_{uff} > 0$ such that the integrand norm in \eqref{eq:z_jump_integral} is bounded by $M_z = \|C_m\|\bigl(M_{\dot{r}} + \|A\| M_x + \|B\| M_u + \|B\| M_{uff}\bigr)$. Consequently:
\begin{equation}
\label{eq:ztk_continue}
    \left\|z(t_k^+) - z(t_k)\right\| \leq \lim_{h \to 0^+} \int_{t_k}^{t_k+h} M_z\,\mathrm{d}\tau = \lim_{h \to 0^+} (h M_z) = 0.
\end{equation}
    This establishes $z(t_k^+) = z(t_k)$. Finally, substituting \eqref{eq:x_jump_continuity} and \eqref{eq:ztk_continue} into \eqref{eq:Cm1} yields $
    s(t_k^+) = C_m e_s(t_k^+) + z(t_k^+) = C_m e_s(t_k) + z(t_k) = s(t_k)$, completing the proof.
\end{proof}

\subsection{Proof of Lemma \ref{lem:finite_time_s}}
\label{appendix:proof_lem:finite_time_s}



\begin{proof}
\textit{Step 1: Finite-time reaching phase.}
When the state is outside the boundary layer ($|s| > \delta$), the saturation function operates in its discontinuous region. From \eqref{eq:rho_gain_condition} and \eqref{eq:sdot_autonomous}, the derivative of $|s(t)|$ evaluates to:
\begin{equation}
\label{eq:ddd_s}
\begin{aligned}
    \frac{d}{dt}|s(t)|  &= \mathrm{sgn}(s)\dot{s} = -\rho - \mathrm{sgn}(s)G d - \mathrm{sgn}(s)C_m w \\
     &\leq -\rho + |G|\Delta + \|C_m\|\Delta_w  \leq -\eta < 0.
\end{aligned}
\end{equation}
Since $\frac{d}{dt}|s(t)| \leq -\eta$ holds pointwise for all $t$ outside $\Omega_\delta$, the function $|s(t)|$ is strictly monotonically decreasing. Integrating \eqref{eq:ddd_s} from $t=0$ to the reaching time $t^*$, where $|s(t^*)| = \delta$, yields:
\begin{equation}
\label{eq:ddd_s_eta}
 \vert{}s(t^*)\vert{} - \vert{}s(0)\vert{}  =   \delta - |s(0)| \leq -\eta t^*.
\end{equation}
Rearranging \eqref{eq:ddd_s_eta} gives the reaching time bound $t^* \leq (|s(0)| - \delta)/\eta$.

\textit{Step 2: Positive invariance of the boundary layer.}
To rigorously establish that the set $\Omega_\delta$ is positively invariant, we evaluate the vector field exactly at the boundary $|s| = \delta$. Considering the Lyapunov candidate $V_s(s) = \frac{1}{2}s^2$, its time derivative at the boundary is:
\begin{equation}
    \dot{V}_s\big|_{|s|=\delta} = s\dot{s} = -\rho \frac{s^2}{\delta} - s(Gd + C_mw) = -\rho\delta - s(Gd + C_mw).
\end{equation}
Bounding the external perturbation terms yields:
\begin{equation}
    \dot{V}_s\big|_{|s|=\delta} \leq -\rho\delta + |s|\bigl(|G|\Delta + \|C_m\|\Delta_w\bigr) = -\delta\bigl(\rho - |G|\Delta - \|C_m\|\Delta_w\bigr).
\end{equation}
Substituting the gain condition $\eta = \rho - |G|\Delta - \|C_m\|\Delta_w > 0$, we obtain:
\begin{equation}
    \dot{V}_s\big|_{|s|=\delta} \leq -\delta\eta < 0.
\end{equation}
Because the derivative $\dot{V}_s$ is strictly negative exactly at the boundary $|s| = \delta$, the state trajectories are forced strictly inward. Consequently, any trajectory that enters $\Omega_\delta$ can never subsequently escape, confirming that $\Omega_\delta$ is a positively invariant set. This completes the proof.
\end{proof}

\subsection{Proof of Lemma \ref{lem:equivalent_control_error}}
\label{appendix:proof_lem equivalent_control_error}
\begin{proof}
For all $t \geq t^*$, the sliding variable satisfies $|s(t)| \leq \delta$. The sliding mode control law is given by:
\begin{equation}
\label{eq:usm_linear}
    u_{sm}(t) = G^{-1}\rho\left( {s(t)}/{\delta}\right).
\end{equation}
Define $v(t) = u_{sm}(t) + d(t)$. From \eqref{eq:usm_linear}, we obtain
\begin{equation}
\label{eq:v_dot_initial}
    \dot{v}(t) = (G^{-1} {\rho}/{\delta})\dot{s}(t) + \dot{d}(t).
\end{equation}
Combining $v(t) = u_{sm}(t) + d(t)$ and \eqref{eq:sdot_autonomous}, we have $\dot{s}(t) = -G(u_{sm}(t) + d(t)) - C_m w(t)$. Substituting this expression into \eqref{eq:v_dot_initial} yields
\begin{equation}
\label{eq:v_dot_expanded}
\begin{aligned}
    \dot{v}(t) &= -\frac{\rho}{\delta}v(t) - G^{-1}\frac{\rho}{\delta}C_m w(t) + \dot{d}(t),
\end{aligned}
\end{equation}
which establishes \eqref{eq:v_differential}. 
Defining the exogenous input $\psi(t) = \dot{d}(t) - G^{-1}\frac{\rho}{\delta}C_m w(t)$, equation \eqref{eq:v_dot_expanded} becomes $\dot{v}(t) = -(\rho/\delta)v(t) + \psi(t)$. Applying the variation of constants formula to this linear system yields the upper bound:
\begin{equation}
\label{eq:|vt|}
    |v(t)| \leq e^{-\frac{\rho}{\delta}(t-t^*)}|v(t^*)| + \frac{\delta}{\rho}  \sup_{\tau \geq t^*} |\psi(\tau)|.
\end{equation}
Following Assumptions \ref{assum:perturbation_noise_bounds} and \ref{assum:disturbance_bound}, we have $\|w(t)\| \leq \Delta_w$ and $|\dot{d}(t)| \leq \Delta_1$. Thus, we obtain
\begin{equation}
    \sup_{\tau \geq t^*} |\psi(\tau)| \leq \Delta_1 + |G|^{-1}\frac{\rho}{\delta}\|C_m\|\Delta_w.
\end{equation}
As $t \to \infty$, the transient exponential term in \eqref{eq:|vt|} decays to zero, yielding the steady-state ultimate bound:
\begin{equation}
    \limsup_{t \to \infty} |v(t)| \leq  \frac{\delta \Delta_1}{\rho} + \frac{\|C_m\|\Delta_w}{|G|},
\end{equation}
which establishes \eqref{eq:v_ultimate_bound} and completes the proof.
\end{proof}

\subsection{Proof of Theorem \ref{thm:continuous_flow_stability}}
\label{appendix:proof_lemma_flow}

\begin{proof}
\textit{Step 1: Lie derivative of the Lyapunov function for the $\chi$-subsystem.} 
During continuous flows, the forced dynamics of the tracking-error subsystem are governed by
\begin{equation}
\label{eq:chi_dynamics_combined}
    \dot{\chi} = A_\chi \chi + E_r(q_r - w) + B_c v,
\end{equation}
where $v = u_{sm} + d$.
Since the RC-ISMC architecture~\eqref{eq:hybrid_interconnected} shares the system matrices $A_\chi$, $E_r$, and $B_c$ with the standalone reset system~\eqref{eq:chi_nominal}, the Lyapunov matrix $P$ and scalar $\lambda_0$ satisfying Assumption~\ref{assum:iss_reset_input} apply directly. Differentiating $V_e(\chi) = \chi^\top P\chi$ along the trajectories of~\eqref{eq:chi_dynamics_combined} and applying Young's inequality—as detailed in the proof of Corollary~\ref{cor:quadratic_iss_reset_sufficient} (Appendix~\ref{appendix:proof_quadratic_iss_reset})—yields:
\begin{equation}
\label{eq:Ve_dot_combined}
\begin{aligned}
    \dot{V}_e &\leq -\frac{\lambda_0}{4} V_e + \frac{4}{\lambda_0}\bigl\|E_r^\top P E_r\bigr\|\|q_r\|^2 + \frac{4}{\lambda_0}\bigl\|E_r^\top P E_r\bigr\|\|w\|^2 + \frac{4}{\lambda_0}\bigl\|B_c^\top P B_c\bigr\||v(t)|^2.
\end{aligned}
\end{equation}
Lemma~\ref{lem:equivalent_control_error} establishes $\limsup_{t\to\infty} |v(t)| \leq \bar{v}_\delta$. Thus, for any $\epsilon > 0$, there exists a finite time $T_\epsilon \geq t^*$ such that $|v(t)| \leq \bar{v}_\delta + \epsilon$ for all $t \geq T_\epsilon$. Substituting this bound, $\|q_r\|^2 \leq \bar{q}_r^2$, and $\|w\|^2 \leq \Delta_w^2$ into \eqref{eq:Ve_dot_combined} gives:
\begin{equation}
\label{eq:v_chi}
    \dot{V}_e(\chi) \leq - ({\lambda_0}/{4}) V_e(\chi) + c_q + c_w + c_{\mathrm{coup},\epsilon}, \quad \forall t \geq T_\epsilon.
\end{equation}

\textit{Step 2: Lie derivative of the Lyapunov function for the sliding variable dynamics.}
For $t \geq t^*$, we have $|s(t)| \leq \delta$ and $u_{sm} = G^{-1}\rho(s/\delta)$. Let $V_s(s) = \frac{1}{2}s^2$. From \eqref{eq:sdot_autonomous}, the time derivative is given by
\begin{equation}
\label{eq: dot_Vs}
    \dot{V}_s = s\dot{s} = -\frac{\rho}{\delta}s^2 - sGd - sC_mw = -\frac{2\rho}{\delta}V_s - sGd - sC_mw.
\end{equation}
Applying Young's inequality to the perturbation terms with $\mu = \rho/(2\delta)$ yields $-sGd \leq \frac{\rho}{4\delta}s^2 + \frac{\delta}{\rho}G^2\Delta^2$ and $-sC_mw \leq \frac{\rho}{4\delta}s^2 + \frac{\delta}{\rho}\|C_m\|^2\Delta_w^2$. Substituting these bounds into \eqref{eq: dot_Vs} gives:
\begin{equation}
\label{eq:vs_dot_bound}
\begin{aligned}
    \dot{V}_s &\leq -\frac{2\rho}{\delta}V_s + \left( \frac{\rho}{2\delta}V_s + \frac{\delta}{\rho}G^2\Delta^2 \right) + \left( \frac{\rho}{2\delta}V_s + \frac{\delta}{\rho}\|C_m\|^2\Delta_w^2 \right) \\
              &= -\frac{\rho}{\delta}V_s + c_s.
\end{aligned}
\end{equation}

\textit{Step 3: Composite flow dissipation.}
Following \cite[Lemma 4.7]{khalil2002nonlinear}, we define the Lyapunov function $V(\xi_e) = V_e(\chi) + V_s(s)$ for the RC-ISMC system. For $t \geq T_\epsilon$, summing \eqref{eq:v_chi} and \eqref{eq:vs_dot_bound} provides the derivative:
\begin{equation}
    \dot{V}(\xi_e) \leq -\frac{\lambda_0}{4}V_e(\chi) - \frac{\rho}{\delta}V_s(s) + c_q + c_w + c_{\mathrm{coup},\epsilon} + c_s.
\end{equation}
Defining the decay rate $\alpha = \min(\lambda_0/4,\, \rho/\delta)$ and factoring out $-\alpha$, we obtain
\begin{equation}
\label{eq:V_composite_final}
    \dot{V}(\xi_e) \leq -\alpha \bigl(V_e(\chi) + V_s(s)\bigr) + c_{\mathrm{total},\epsilon} = -\alpha V(\xi_e) + c_{\mathrm{total},\epsilon},
\end{equation}
where $c_{\mathrm{total},\epsilon} = c_q + c_w + c_{\mathrm{coup},\epsilon} + c_s$, which establishes \eqref{eq:global_flow_bound}.

\textit{Step 4: Non-increase of the Lyapunov function at resets.}
Evaluating the Lyapunov function at the reset instant $t_k^+$ using \eqref{eq:jump_components}, \eqref{eq:interconnected_jump_map}, and then applying the condition in \eqref{eq:tracking_jump_lmi}, gives:
\begin{equation}
\label{eq:jump_nonexpansive_total}
\begin{aligned}
    V(\xi_e(t_k^+)) &= \chi(t_k)^\top \bigl( A_{\chi,\rho}^\top P A_{\chi,\rho} \bigr) \chi(t_k) + V_s(s(t_k)) \\
                    &\leq \chi(t_k)^\top P \chi(t_k) + V_s(s(t_k)) = V(\xi_e(t_k)).
\end{aligned}
\end{equation}

\textit{Step 5: Ultimate boundedness conclusion.}
Evaluating \eqref{eq:V_composite_final} and \eqref{eq:jump_nonexpansive_total} via the comparison lemma \cite{khalil2002nonlinear} on $[T_\epsilon, \infty)$ provides the asymptotic bound $\limsup_{t \to \infty} V(\xi_e(t)) \leq c_{\mathrm{total},\epsilon}/\alpha$. Applying Lemma~\ref{lem:equivalent_control_error} and taking the limit as $\epsilon \to 0$ yields
\begin{equation}
    \limsup_{t \to \infty} V(\xi_e(t)) \leq {c_{\mathrm{total}}}/{\alpha}.
\end{equation}
Because the Lyapunov function satisfies the lower bound $V(\xi_e) \geq \lambda^\ast \|\xi_e\|^2$, isolating $\|\xi_e\|$ directly derives the state limit in~\eqref{eq:ultimate_bound}. This bounded response to external perturbations establishes that the system is ISS and UUB. This completes the proof.
\end{proof}

\end{document}